\documentclass{article}
\usepackage{ijcai22}

\usepackage{times}

\usepackage{soul}
\usepackage{url}
\usepackage[hidelinks]{hyperref}
\usepackage[utf8]{inputenc}
\usepackage[small]{caption}
\usepackage{graphicx}
\usepackage{amsmath}
\usepackage{booktabs}
\usepackage{amsfonts}
\usepackage{amsmath}
\usepackage{amssymb}
\usepackage{amsthm}
\usepackage{colortbl}
\usepackage{bm}

\usepackage{enumitem}

\usepackage{xspace}
\usepackage{fancyhdr}
\usepackage{tcolorbox}

\usepackage{todonotes}

\newif\iflong
\newif\ifshort

\longtrue

\iflong
\else
\shorttrue
\fi

\newtheorem{theorem}{Theorem}
\newtheorem{lemma}[theorem]{Lemma}
\newtheorem{corollary}[theorem]{Corollary}
\newtheorem{proposition}[theorem]{Proposition}

\newtheorem{observation}{Observation}

\theoremstyle{definition}

\newcommand{\cc}[1]{{\mbox{\textnormal{\textsf{#1}}}}\xspace}  

\newcommand{\NP}{\cc{NP}}

\newcommand{\FPT}{\cc{FPT}}
\newcommand{\XP}{\cc{XP}}
\newcommand{\Weft}{{\cc{W}}}
\newcommand{\W}[1]{{\Weft}{\normalfont{[#1]}}}

\newcommand{\tw}{\operatorname{tw}}

\newcommand{\Nat}{\mathbb{N}}
\newcommand{\bigoh}{\mathcal{O}}

\newcommand{\Schelling}{\textsc{Schelling}\xspace}
\newcommand{\SchellingM}{\textsc{SchellingM}\xspace}
\newcommand{\clique}{\textsc{Clique}\xspace}
\newcommand{\biclique}{\textsc{BiClique}\xspace}
\newcommand{\minbisection}{\textsc{MinBisection}\xspace}
\newcommand{\UBP}{\textsc{UnBinPacking}\xspace}

\newcommand{\vbf}{\ensuremath{\mathbf{v}}\xspace}
\newcommand{\xbf}{\ensuremath{\mathbf{x}}\xspace}
\newcommand{\ybf}{\ensuremath{\mathbf{y}}\xspace}
\newcommand{\ubf}{\ensuremath{\mathbf{u}}\xspace}

\newcommand{\sw}{\ensuremath{\text{SW}}\xspace}
\newcommand{\WO}{\ensuremath{\text{WO}}\xspace}
\newcommand{\PO}{\ensuremath{\text{PO}}\xspace}
\newcommand{\UVO}{\ensuremath{\text{UVO}}\xspace}
\newcommand{\GWO}{\ensuremath{\text{GWO}}\xspace}
\newcommand{\perfect}{\ensuremath{\text{Perfect}}\xspace}

\newcommand{\tuple}[1]{\ensuremath{\langle {#1} \rangle}\xspace}

\newcommand{\Sizes}{\operatorname{Sizes}}
\newcommand{\Types}{\chi\operatorname{-Types}}
\newcommand{\NeighborhoodSizes}{\chi\operatorname{-Neighbors}}

\title{The Parameterized Complexity of Welfare Guarantees in Schelling Segregation}

\author{ 
Argyrios Deligkas$^1$ \and
Eduard Eiben$^1$\and
Tiger-Lily Goldsmith$^{1}$\\
\affiliations
$^1$Royal Holloway, University of London\\
\emails
\{Argyrios.Deligkas, Eduard.Eiben\}@rhul.ac.uk,
 tigerlilygoldsmith@gmail.com
}

\begin{document}

\maketitle

\begin{abstract}
Schelling's model considers $k$ types of agents each of whom needs to select a vertex on an undirected graph, where every agent prefers to neighbor agents of the same type. 
We are motivated by a recent line of work that studies solutions that are optimal with respect to notions related to the welfare of the agents.
We explore the parameterized complexity of computing such solutions.
We focus on the well-studied notions of social welfare (\WO) and Pareto optimality (\PO), alongside the recently proposed notions of group-welfare optimality (\GWO) and utility-vector optimality (\UVO), both of which lie between \WO and \PO.
Firstly, we focus on the fundamental case where $k=2$ and there are $r$ red agents and $b$ blue agents. We show that all solution-notions we consider are \NP-hard to compute even when $b=1$ and that they are \W{1}-hard when parameterized by $r$ and $b$. 
In addition, we show that \WO and \GWO are \NP-hard even on cubic graphs. 
We complement these negative results by an \FPT algorithm parameterized by $r, b$ and the maximum degree of the graph. 
For the general case with $k$ types of agents, we prove that for any of the notions we consider the problem is $\W{1}$-hard when parameterized by $k$ \iflong, and thus \NP-hard when $k$ is a part of the input, \fi
for a large family of graphs that includes trees. We accompany these negative results with an \XP algorithm parameterized by $k$ and the treewidth of the graph.
\end{abstract}

\section{Introduction}
\label{sec:intro}
Residential segregation is a phenomenon that is observed in many residential areas around the globe.
As a result of de-facto segregation, people group together forming communities based on traits such as race, ethnicity, socioeconomic status and residential areas become noticeably divided into segregated neighborhoods. Half a century ago, \cite{schelling1969models} proposed a simple agent-based model to address residential segregation and study how segregation emerges from individuals' perceptions. 

At a high level, Schelling's model works as follows. There are two types of agents, say red and blue, each of whom is placed on a unique node on a graph. Agents are aware of their neighborhood; agents of the same type are considered ``friends'' and those of opposite type ``enemies''. An agent is happy with their location if and only if the fraction of friends in their neighborhood is at least $\tau$, where $\tau \in [0,1]$ is a tolerance parameter. Schelling proposed a random process that starts from a random initial assignment and agents who are unhappy in their current neighborhood relocate to a different, random, empty node, whilst happy agents stay put. It is expected that when agents are not tolerant towards a diverse neighborhood, $\tau>\frac{1}{2}$, these dynamics will converge to a segregated assignment. However, Schelling's experimentation on grid graphs showed that even when agents are in favour of integration, i.e. $\tau \approx \frac{1}{3}$, the  final assignment will be segregated.

Since Schelling's model was proposed, his work has been the subject of many empirical studies in sociology~\cite{clark2008understanding}, in economics~\cite{zhang2004dynamic,zhang2004residential}, and more recently in computer science. For example~\cite{DBLP:journals/corr/BarmpaliasEL15} and~\cite{immorlica2017exponential}, analyze Schelling's model on a grid graph with its original random dynamics, as well as many variants of this random process. They show that assignments converge to large monochromatic subgraphs with a high probability, confirming Schelling's research. Even more recently, \cite{BSV21} studied assignments with certain {\em welfare guarantees} for the agents and the computational complexity of computing them. These guarantees are the focus of this paper, albeit under the prism of Parameterized Complexity.

In parameterized
algorithmics~\iflong\cite{CyganFKLMPPS15,DowneyFellows13,Niedermeier06}\fi\ifshort\cite{CyganFKLMPPS15}\fi{} the
running-time of an algorithm is studied with respect to a parameter
$k\in\Nat_0$ and input size~$n$. The basic idea is to find a parameter
that describes the structure of the instance such that the
combinatorial explosion can be confined to this parameter. In this
respect, the most favorable complexity class is \FPT
(\textit{fixed-parameter tractable}), which contains all problems that
can be decided by an algorithm running in time $f(k)\cdot
n^{\bigoh(1)}$, where $f$ is a computable function. Algorithms with
this running-time are called \emph{fixed-parameter (FPT) algorithms}. A less
favorable, but still positive, outcome is an \XP{} \emph{algorithm}, which is an algorithm
running in time $\bigoh(n^{f(k)})$; problems admitting such
algorithms belong to the class \XP. Finally, showing that a problem is $\W{1}$-hard rules out the existence of a fixed-parameter algorithm under the well-established assumption that $\W{1}\neq \FPT$.

\subsection{Our Contributions}
We explore the parameterized complexity of computing assignments for Schelling's model that optimize some welfare guarantee. We study four solution notions: social-welfare optimality (\WO), Pareto optimality (\PO), group-welfare optimality (\GWO) and utility-vector optimality (\UVO). We denote the problem as $\phi$-\Schelling, where $\phi \in \{ \WO, \PO, \GWO, \UVO\}$, and the task is to find an assignment that satisfies notion $\phi$ for a given Schelling instance.
While \WO and \PO are well-studied notions in various domains, the solution concepts of \GWO and \UVO were proposed by~\cite{BSV21}. There it was proven that both \UVO and \GWO lie between \WO and \PO. At a high level, an assignment is \GWO if we cannot increase the total utility of one type of agents without decreasing the utility of the other type; an assignment is \UVO if it is not possible to improve the sorted utility vector of the agents. While~\cite{BSV21} showed that all four notions are \NP-hard to compute in general, their parameterized complexity remained open.

We firstly focus on the fundamental case where we have two types of agents: $r$ red agents and $b$ blue agents. In Theorem~\ref{thm:b-np} we show that $\phi$-\Schelling is \NP-hard  even when $b=1$, for every $\phi \in \{\WO, \PO, \GWO, \UVO \}$. In Theorem~\ref{thm:rb-w-h} we extend this negative result and we show that deciding if there exists a {\em perfect} assignment, i.e. an assignment where every agent has only friends as neighbors, is \W{1}-hard when parameterized by $r+b$. This implies Corollary~\ref{cor:rb-wh}: $\phi$-\Schelling is \W{1}-hard when parameterized by $r+b$, for every $\phi \in \{\WO, \PO, \GWO, \UVO \}$. Hence, if we want to derive a positive result, we {\em need} to restrict the topology of the graph. In Theorem~\ref{thm:bounded-NP-h} we show that restricting the maximum degree of the graph does not always suffice; we prove that both \WO-\Schelling and \GWO-\Schelling are \NP-hard even on cubic graphs. 
We complement these negative results by Theorem~\ref{thm:fpt-rbD}; we show that $\phi$-\Schelling is in \FPT, for all four optimality notions, parameterized by $r+b+\Delta$, where $\Delta$ denotes the maximum degree of the graph. In fact, we show that $\phi$-\Schelling admits a polynomial time preprocessing algorithm, called \emph{kernel}, that yields an instance with at most \(\mathcal{O}(\Delta^2\cdot r^2\cdot b^2)\) many vertices.

Then, we turn our attention to the general case where there are {\em multiple} types of agents, which we denote \SchellingM. In Theorem~\ref{thm:wh-k-types} we prove that finding a perfect assignment is $\W{1}$-hard when parameterized by the types of agents, $k$, for a large family of graphs that includes trees. Therefore, it is \NP-hard when $k$ is part of the input and not bounded by a function of the parameter for the same family of graphs. Again, we get the corresponding \NP-hardness and $\W{1}$-hardness for $\phi$-\SchellingM as corollaries, for every $\phi \in \{\WO, \PO, \GWO, \UVO \}$.
We complement this with three positive results. In Theorem~\ref{thm:tw_algorithm} we derive an \XP algorithm parameterized by the number of types and the treewidth of the graph. By using the same algorithm, we get Corollary~\ref{cor:fpt-tw-n} that shows an FPT algorithm for $\phi$-\SchellingM parameterized by the number of agents plus the treewidth of the graph. Finally, by slightly modifying this algorithm, we get Corollary~\ref{cor:fpt-perfect} that shows that if the number of types is any fixed constant, then the problem of finding a perfect assignment, if one exists, admits an FPT algorithm parameterized by treewidth. 

\subsection{Further Related Work}
A different line of work studies Schelling {\em games}, a strategic setting of Schelling's model. There, unhappy agents will move to a different position that maximises the fraction of friends in the neighborhood. Here the focus is shifted to the existence of Nash equilibria, i.e., assignments where no agent has incentives to change their position.
In~\cite{ElkindGISV19} they consider {\em jump} Schelling games with $k \geq 2$ types, with agents that can deviate to empty nodes in the graph and {\em stubborn} agents which do not move regardless of their utility. They proved \NP-hardness for computing a Nash equilibrium and for \WO. In~\cite{AEGV20} {\em swap} Schelling games were studied, where agents of different types exchange their positions if at least one of them strictly increase their utility. Again, they showed that deciding whether a Nash equilibrium exists, is \NP-hard. Furthermore, in order to measure the diversity in assignments, they introduced the {\em degree of integration} that counts the number of agents exposed to agents not of their type. They showed that computing assignments that maximize this measure is hard.

\cite{DBLP:conf/mfcs/BiloBLM20} investigate the existence of equilibria via finite improvement paths on different graph classes for swap Schelling games, and study a local variant wherein agents can only swap with agents in their neighborhood. 

\cite{kanellopoulos2020modified} study price of anarchy and price of stability in {\em modified} Schelling games, where the agent includes herself as part of the neighborhood\iflong (effectively adding $+1$ to the denominator of the standard utility function). Hence, this captures an agents desire to be around a larger neighborhood of friends\fi. They prove tight bounds on the price of anarchy for general and some specific graphs with $k\geq2$ and $k=1$.
  
 Furthermore, there are other extensions and variations of Schelling games~\cite{KKVstrangers,echzell2019convergence,chauhan2018schelling}.

\ifshort
\smallskip

\noindent {\emph{Statements where proofs or details are omitted due to
space constraints are marked with $\star$. A version
containing all proofs and details is provided as supplementary material.}}
\fi

\section{Preliminaries}
\label{sec:prelims}
For every positive integer $n$, let $[n] = \{1, 2, \ldots, n\}$. Given two vectors $\xbf, \ybf$ of length $n$, we say that \xbf~{\em weakly dominates} \ybf if $\xbf(i) \geq \ybf(i)$ for every $i \in [n]$; \xbf~{\em strictly dominates} \ybf if at least one of the inequalities is strict.

A {\em Schelling instance} \tuple{G,A}, consists of a graph $G=(V,E)$ and a set of {\em agents} $A$, where $|A| \leq |V|$. Every agent has a {\em type}, or {\em color}. When there are only two colors available, we assume that $A = R \cup B$, where $R$ contains red agents and $B$ contains blue agents. We denote $r = |R|$ and $b = |B|$. Agents $i$ and $j$ are {\em friends}, if they have the same color; otherwise they are {\em enemies}. For any agent $i$ we use $F(i)$ to declare the set of his friends.

An {\em assignment} $\vbf = (v(1), \ldots, v(|A|))$ for the Schelling instance \tuple{G,A} maps every agent in $A$ to vertex $v \in V$, such that every vertex is occupied by at most one agent. Here, $v(i) \in V$ is the vertex of $G$ that agent $i$ occupies. For any assignment \vbf and any agent $i \in A$, $N_i(\vbf) = \{ j \in A: v(i)v(j) \in E\}$ denotes the set of neighbors of $v(i) \in V$ that are occupied under \vbf. Let $f_i(\vbf) = |N_i(\vbf) \cap F(i)|$ and let $e_i(\vbf) = |N_i(\vbf)| - f_i(\vbf)$ be respectively the numbers of neighbors of agent $i$ who are his friends and his enemies under \vbf. The {\em utility} of agent $i$ under assignment \vbf, denoted $u_i(\vbf)$, is 0 if $|N_i(\vbf)|=0$, and if $|N_i(\vbf)| \neq 0$ is defined as
\begin{align*}
    u_i(\vbf) = \frac{f_i(\vbf)}{|N_i(\vbf)|} = \frac{f_i(\vbf)}{f_i(\vbf) + e_i(\vbf)}.
\end{align*}
The {\em social welfare} of \vbf is the sum of the utilities of all agents, formally $\sw(\vbf) = \sum_{i \in A}u_i(\vbf)$. For $X \in \{R, B\}$ we denote $\sw_X(\vbf) = \sum_{i \in X}u_i(\vbf)$. 

We use $\ubf(\vbf)$ to denote the vector of length $|A|$ that contains the utilities of the agents under \vbf, sorted in non-increasing order. Similarly, let $\ubf_X(\vbf)$ denote the corresponding vector of utilities of the agents in $X \in \{R, B\}$. An assignment \vbf is  {\em utility-vector dominated} by $\vbf'$ if $\ubf(\vbf')$ strictly dominates $\ubf(\vbf)$; \vbf is {\em group-welfare dominated} by $\vbf'$ if $\sw_X(\vbf') \geq \sw_X(\vbf)$, where $X \in \{R,B\}$, and at least one of the inequalities is strict. An assignment \vbf is:
\begin{itemize}
    \item {\em welfare optimal}, denoted \WO, if for every other assignment $\vbf'$ we have $\sw(\vbf) \geq \sw(\vbf')$;
    \item {\em Pareto optimal}, denoted \PO, if and only if there is no $\vbf'$  such that $\ubf_X(\vbf')$ weakly dominates $\ubf_X(\vbf)$ for $X \in \{R, B\}$ and at least one of the dominations is strict;
    \item {\em utility-vector optimal}, denoted \UVO, if it is not utility-vector dominated by any other assignment;
    \item {\em group-welfare optimal}, denoted \GWO, if it is not group-welfare dominated by any other assignment;
    \item {\em perfect}, denoted \perfect, if every agent gets utility 1.
\end{itemize}
As previously mentioned, \UVO and \GWO were introduced in~\cite{BSV21} where the following proposition was proven.
\begin{proposition}
\label{pro:notions}
If an assignment \vbf is \WO, then it is \UVO, \GWO,~and~\PO. If \vbf is \UVO or \GWO, then it is \PO.
\end{proposition}
\begin{observation}
\label{obs:perfect}
If Schelling instance \(\tuple{G,A}\) admits a \perfect assignment, then every \PO assignment is \perfect. 
\end{observation}

\noindent
In this paper we study the complexity of $\phi$-\Schelling, where $\phi \in \{\WO, \PO, \GWO, \UVO, \perfect \}$. In other words, given a Schelling instance \tuple{G,A}, we study the problem of finding an assignment \vbf satisfying the given optimality notion.

\subsection{Parameterized Complexity}

\paragraph{\bf Parameterized Complexity.}
We refer to the handbook by Diestel~\shortcite{Diestel12} for
standard graph terminology. We also refer to the standard books for a basic overview of parameterized complexity theory~\iflong\cite{CyganFKLMPPS15,DowneyFellows13}\fi\ifshort\cite{CyganFKLMPPS15}\fi, and assume that readers are aware of the complexity classes \FPT, \XP\ and \W{1}. Readers interested in the full details of the proof of Theorem~\ref{thm:tw_algorithm} are also expected to have a basic understanding of \emph{treewidth} and \emph{nice tree-decompositions}~\cite{CyganFKLMPPS15}.
\iflong
We denote by $\Nat$ the set of natural numbers, by $\Nat_0$ the set $\Nat \cup \{0\}$. 
We refer to the handbook by Diestel~\shortcite{Diestel12} for
standard graph terminology. Let $K_{i,j}$ be the complete bipartite graph with parts of size $i$ and $j$.


A parameterized optimization problem is {\em kernelizable}
if there exists a polynomial-time  \emph{preprocessing algorithm} that maps an instance $(I, \kappa)$ of
the problem to another instance $(I', \kappa')$ such that 
$|I'| \leq f(\kappa)$ and $\kappa' \leq f(\kappa)$, where $f$ is a computable function called \emph{size of the kernel}, and 
a polynomial-time \emph{solution lifting algorithm} that takes as an input the instance $(I, \kappa)$, the output of the preprocessing algorithm $(I', \kappa')$, and a solution $s'$ for $(I', \kappa')$ and computes a solution $s$ for $(I, \kappa)$.
The instance
$(I',\kappa')$ is called the {\em kernel} of~$I$. It is well known that a
decidable problem is \FPT{} if and only if it is
kernelizable~\cite{DowneyFellows13}.
A \emph{polynomial kernel} is a kernel whose size can be bounded by a
polynomial in the parameter.

\smallskip
\noindent \textbf{Treewidth.}\quad
A \emph{nice tree-decomposition}~$\mathcal{T}$ of a graph $G=(V,E)$ is a pair 
$(T,\chi)$, where $T$ is a tree (whose vertices we call \emph{nodes}) rooted at a node $r$ and $\chi$ is a function that assigns each node $t$ a set $\chi(t) \subseteq V$ such that the following holds: 
\begin{itemize}[noitemsep]
	\item For every $uv \in E$ there is a node
	$t$ such that $u,v\in \chi(t)$.
	\item For every vertex $v \in V$,
	the set of nodes $t$ satisfying $v\in \chi(t)$ forms a subtree of~$T$.
	\item $|\chi(\ell)|=0$ for every leaf $\ell$ of $T$ and $|\chi(r)|=0$.
	\item There are only three kinds of non-leaf nodes in $T$:
	\begin{itemize}[noitemsep,label=]
        \item \textbf{Introduce node:} a node $t$ with exactly
          one child $t'$ such that $\chi(t)=\chi(t')\cup
          \{v\}$ for some vertex $v\not\in \chi(t')$.
        \item \textbf{Forget node:} a node $t$ with exactly
          one child $t'$ such that $\chi(t)=\chi(t')\setminus
          \{v\}$ for some vertex $v\in \chi(t')$.
        \item \textbf{Join node:} a node $t$ with two children $t_1$,
          $t_2$ such that $\chi(t)=\chi(t_1)=\chi(t_2)$.
	\end{itemize}
\end{itemize}

The \emph{width} of a nice tree-decomposition $(T,\chi)$ is the size of a largest set $\chi(t)$ minus~$1$, and the \emph{treewidth} of the graph $G$,
denoted $\tw(G)$, is the minimum width of a nice tree-decomposition of~$G$.
Efficient fixed-parameter algorithms are known for computing a nice tree-decomposition of near-optimal width~\cite{BodlaenderDDFLP16,Kloks94}. Whenever we speak of the treewidth of a directed graph, we mean the treewidth of its underlying undirected graph.

We let $T_t$ denote the subtree of $T$ rooted at a node $t$, and we use $\chi(T_t)$ to denote the set $\bigcup_{t'\in V(T_t)}\chi(t')$ and \(G_t\) to denote the graph \(G[\chi(T_t)]\) induced by the vertices in \(\chi(T_t\).

\begin{proposition}[\citeauthor{BodlaenderDDFLP16}, \citeyear{BodlaenderDDFLP16}]\label{fact:findtw}%
	There exists an algorithm which, given an $n$-vertex graph $G$ and an integer~$k$, in time $2^{\bigoh(k)}\cdot n$ either outputs a tree-decomposition of $G$ of width at most $5k+4$ and $\bigoh(n)$ nodes, or determines that $\tw(G)>k$.
\end{proposition}  
\fi

\section{Parameterizing by $r$ and $b$}
\label{sec:agents}
In this section, we study $\phi$-\Schelling parameterized by the number of red and blue agents.
We firstly focus on the number of blue agents, $b$. Observe that in this case, if $r+b = |V|$, there is a trivial \XP algorithm since there are ${n \choose{b}} = \bigoh(|V|^b)$ assignments in total; for any choice of the positions of the $b$ blue agents, the remaining vertices have to be occupied by red agents. This \XP algorithm is the best we can hope for;~\citeauthor{BSV21}, although they do not mention it, show $\W{1}$-hardness for \WO-\Schelling parameterized by $b$.

The above-mentioned \XP algorithm works because we can trivially extend a choice for the positions of the blue agents to a complete assignment; there are no choices to be made for red agents. This is no longer possible when $r+b < |V|$. For this case, ~\citeauthor{ElkindGISV19} showed that \WO-\Schelling is \NP-hard even when $b=1$. However, their proof was relying on the assumption that the blue agent is ``stubborn'', i.e. the blue agent had a fixed position on the graph. We strengthen their result by showing that the problem remains \NP-hard, even when the blue agent is not stubborn.
\iflong
\begin{theorem}
\fi
\ifshort
\begin{theorem}[$\star$]
\fi
\label{thm:b-np}
If $r+b < |V|$, then $\phi$-\Schelling is \NP-hard, for $\phi \in \{\WO,\PO,\UVO,\GWO\}$, even when $b = 1$.
\end{theorem}
\iflong
\begin{proof}
\fi
\ifshort
\begin{proof}[Proof sketch]
\fi
We will prove hardness via a reduction from \clique,
 where we are given a graph $H$ and an integer $k$ and the goal is to decide the existence of a set $S\subseteq V(H)$, where $|S|=k$, such that $H[S]$ induces a clique\iflong, that is, there is an edge between every pair of vertices in $H[S]$\fi. We prove the theorem for \(\phi = \PO\), which implies the hardness for the remaining concepts by Proposition~\ref{pro:notions}. 
 Given an instance \tuple{H, k} of \clique, where $H = (V',E')$ and $|V'| = n$, we construct an instance of \PO-\Schelling as follows:
\begin{itemize}
\item There are $n^2 + k$ red agents and one blue agent. 
\item $G'$ is a clique of size $n^2$, where $G'=(X,Y)$.
\item The topology $G=(V,E)$ is defined so that $V = x \cup V' \cup X$ and $E= E' \cup \{\ xv : v \in V'\} \cup \{\ vw : v \in V', w \in X\}$.
\end{itemize}

\newcommand{\ubfClique}{\ubf_{\mathrm{PO}}}
Note that since $b=1$, it follows that for every assignment \(\vbf\) we get \(\ubf_B(\vbf) = (0)\).
We show that \(H\) admits a clique of size \(k\) if and only if the utility-vector for red agents of every \PO assignment \vbf is equal to the vector \(\ubfClique = (1,\ldots, 1, \frac{n^2 + k -1}{n^2 + k}, \ldots,\frac{n^2 + k -1}{n^2 + k})\), where there are $n^2$ red agents with utility \(1\) and \(k\) red agents with utility \(\frac{n^2 + k -1}{n^2 + k}= 1-\frac{1}{n^2 + k}\).

\ifshort
The correctness follows by first observing that at any \PO assignment the blue agent occupies vertex $x$. Given that, the utility of the red agents is maximized when $n^2$ red agents occupy $G'$ and the positions of the remaining $k$ red agents form a clique in $G$. 
\fi
\iflong
First, assume $H$ has a clique of size $K$. We create an assignment \vbf for \tuple{G,A} such that $n^2$ red agents are on vertices in $G'$ and the remaining $k$ red agents are assigned to the clique in $H$. The blue agent is assigned to vertex $x$. Observe that, since there is no edge between $G'$ and $x$, all agents in $G'$ are only connected to agents of their type. No agents in $G'$ are isolated because it is clique of size $n^2$. Hence, $u_i(\vbf) = 1$ for every agent $i \in G'$. For the remaining $k$ red agents, observe that they are connected to every other red agent and the blue agent. Hence, $u_i(\vbf) = \frac{n^2 + k -1}{n^2 + k}$ for every agent $i \in H$. Therefore, $\ubf(\vbf) = \ubfClique$. Thus, if there exists a clique of size $k$ in G then an assignment \vbf with the utility vector equal to \(\ubfClique\) exists. We need to show that if such assignment \vbf with \(\ubf_R(\vbf) = \ubfClique\) exists, then it is necessarily \PO and there exists a clique of size \(k\) in \(H\).

Let \(\vbf\) be an assignment in \tuple{G,A}. We show that either \(\ubfClique\) strictly dominates \(\ubf_R(\vbf)\), or \(\ubf_R(\vbf)=\ubfClique\) and the graph \(H\) admits a clique of size \(k\). 
Observe that if the blue agent is not at vertex $x$, but instead at a vertex in $V(H)\cup X$, then minimum $n^2 + k - (n-1)$ of the red agents lose utility. Each of these red agents lose minimum $\frac{1}{n^2+k}$ of their utility due to the blue agent becoming part of their neighborhood and the fact that there are \(n^2+k+1\) agents in total. It follows that in this case \(\ubfClique\) strictly dominates \(\ubf_R(\vbf)\). Else, the blue agent is assigned the vertex $x$ by \vbf. Now, every red agent at a vertex in \(H\) loses at least $\frac{1}{n^2+k}$ of their utility. It follows that exactly \(k\) red agents are assigned a vertex in \(H\), otherwise \(\ubfClique\) strictly dominates \(\ubf_R(\vbf)\). 
Finally, if a red agent $i$ at a vertex in $H$ is not adjacent to all the red agents, then the loss of utility of agent \(i\) is at least $\frac{1}{n^2+k-1}> \frac{1}{n^2+k}$. Therefore, if the $k$ red agents at the vertices of $H$ do not form a clique, then \(\ubfClique\) strictly dominates \(\ubf_R(\vbf)\). 
It follows that if \(\ubfClique\) does not strictly dominates \(\ubf_R(\vbf)\), then the $k$ agents that are assigned the vertices of $H$ by the assignment \(\vbf\) form a clique and \(\ubf_R(\vbf)=\ubfClique\).  
\fi
\end{proof}

\noindent On the positive side, we can easily get an \XP algorithm for \perfect-\Schelling parameterized by $b$.
\iflong
\begin{theorem}
\fi
\ifshort
\begin{theorem}[$\star$]
\fi
\label{thm:perfect-xp-b}
For \perfect-\Schelling there is an \XP-algorithm parameterized by $b$.
\end{theorem}
\iflong
\begin{proof}
Observe that when $b=1$, then trivially there is no \perfect assignment since the unique blue agent cannot get utility 1. Now, for $b>1$ we proceed as follows. We guess the $b$ vertices, denoted $B$, that blue agents occupy under the constraint that every connected component induced by $B$ has size at least 2. There are $O(|V|^b)$ such many guesses. Then, we consider the graph induced by $S := V-B-N(B)$, where $N(B)$ contains all the vertices adjacent to at least one vertex in $B$. We focus on the connected components induced by $S$ that have size at least 2; let us denote this graph $G'$. If $G'$ contains less than $r$ vertices, then we reject the guess. Otherwise, we start assigning red agents to the vertices of $G'$ in the following manner. We order the connected components of $G'$ in decreasing size. Then, greedily we start assigning red agents to the vertices of each component under the constraint that every new agent we assign to the current connected component is either the first agent assigned to a vertex of this component, or he is adjacent to a vertex already occupied by an agent. When we will consider the last agent, there are two cases.
\begin{itemize}
    \item The last agent is assigned to a vertex adjacent to a vertex with a red agent. In this case, every red agent has at least one red neighbor and no blue neighbors and every blue agent has a blue neighbor and no red neighbors. Thus, the assignment is \perfect.
    \item The last agent $a$ is the first agent assigned to a vertex of a connected component of $G'$. Then, we check the largest connected component of $G'$. If there is a red agent that can be reassigned to a vertex adjacent to $a$, while the remaining red agents of his previous component still have utility 1, then we make the reassignment and we have a \perfect assignment. Otherwise, we can conclude that the original guess of $B$ cannot be extended to a \perfect assignment and we proceed to the next guess.
\end{itemize}
Hence, in time $O(|V|^b)$ we can decide if a \perfect assignment exists and if it does compute one in the same time.
\end{proof}
\fi

\noindent Next, we show that the \XP-algorithm from Theorem~\ref{thm:perfect-xp-b} is actually the best we can hope for. In fact we show that the problem is hard even if we parameterize by $r+b$.
\iflong
\begin{theorem}
\fi
\ifshort
\begin{theorem}[$\star$]
\fi
\label{thm:rb-w-h}
\perfect-\Schelling is $\W{1}$-hard when parameterized by $r+b$.
\end{theorem}
\iflong
\begin{proof}
\fi
\ifshort
\begin{proof}[Proof sketch]
\fi
We will prove the theorem via a reduction from \biclique. The input for an instance of \biclique is a bipartite graph $H=(L\cup R, Y)$ and an integer $k$. The task is to decide whether $H$ contains a complete bipartite subgraph with $k$ vertices in each side of $H$. It is known that \biclique is $\W{1}$-hard when parameterized by $k$~\cite{Lin2014}. Given $H=(L\cup R, Y)$ and $k$, we construct an instance \tuple{G,A} of \perfect-\Schelling as follows. The graph $G=(V,E)$ will be the complement of $H$. This means that $V=L\cup R$ and $uv \in E$ if and only if $uv \notin Y$. Furthermore, we create $k$ red and $k$ blue agents, i.e. $r=b=k$. We will ask if there is a perfect assignment \vbf.

\ifshort
The correctness follows from the following two facts for any perfect assignment: all the red agents have to be on the left side of the graph and all the blue agents on the other; no agent is adjacent to an enemy.
\fi
\iflong
Firstly, assume that $\{l_1, \ldots, l_k\} \in L$ and $\{r_1, \ldots, r_k\} \in R$ form a complete bipartite subgraph of $H$. We create an assignment \vbf for \tuple{G,A}, by assigning a blue agent to vertex $l_i$ and a red agent to vertex $r_i$ for every $i \in [k]$. Observe that since $l_i$ is adjacent to all $r_1, \ldots, r_k$ in $H$, it follows that in $G$, which recall is the dual of $H$, the vertex $l_i$ is not adjacent to any of the vertices $r_1, \ldots, r_k$. Hence, there is no edge $uv \in E$ such that $u$ is occupied by a red agent and $v$ is occupied by a blue agent. In addition, observe that in $G$ the vertices $l_1, \ldots, l_k$ induce a clique and the vertices $r_1, \ldots, r_k$ induce a clique as well. Hence, under \vbf every blue agent is adjacent to all remaining blue agents and no red agent, and every red agent is a neighbor to all remaining red agents and no blue agent. This means that \vbf is perfect. 

For the other direction, assume that there exists a perfect assignment \vbf in \tuple{G,A}. 
Hence, no red agent has a blue agent as a neighbor in \vbf; if this was the case, then both agents would get utility strictly less than 1. So, assume that under \vbf, $v \in L$ is occupied by a blue agent. Then there is no vertex $u \in L$ occupied by a red agent under \vbf; this is because $L$ forms a clique in $G$. Thus, all red agents occupy vertices of $R$ and no blue agent occupies a vertex in $R$; this is because $R$ forms a clique in $G$. Hence, we can conclude that in \vbf:
\begin{itemize}
    \item the blue agents occupy the vertices $l_1, l_2, \ldots, l_k$ in $L$;
    \item the red agents occupy the vertices $r_1, r_2, \ldots, r_k$ in $R$;
    \item there is no edge $l_ir_j$, where $i \in [k]$ and $j \in [k]$.
\end{itemize}
Hence, in the dual of $G$, which is $H$, for every $i \in [k]$ and every $j \in [k]$ the edge $l_ir_j$ exists. Thus, $\{l_1, \ldots, l_k\}$ and $\{r_1, \ldots, r_k\}$ form a solution of \biclique.
\fi
\end{proof}

The combination of Theorem~\ref{thm:rb-w-h}, Proposition~\ref{pro:notions}, and Observation~\ref{obs:perfect}, gives us the following corollary.
\begin{corollary}
\label{cor:rb-wh}
$\phi$-\Schelling is $\W{1}$-hard when parameterized by $r+b$, for $\phi \in \{\WO, \PO, \UVO, \GWO\}$.
\end{corollary}

\section{Bounded Degree Graphs}
\label{sec:bounded-d}
In light of the negative results from the previous section, we turn our attention on instances where the structure of $G$ is restricted. In this section, we focus on the maximum degree $\Delta$ of $G$. We prove that \WO-\Schelling and \GWO-\Schelling are \NP-hard even when $G=(V,E)$ is cubic, i.e. every vertex has degree 3.

\begin{theorem}
\label{thm:bounded-NP-h}
\WO-\Schelling and \GWO-\Schelling on cubic graphs are \NP-hard.
\end{theorem}
\begin{proof}
It follows from Proposition~\ref{pro:notions} that it suffices to show that \GWO-\Schelling on cubic graphs is \NP-hard.
The proof is via a reduction from \minbisection on cubic graphs~\cite{bui1987graph}. An instance of \minbisection consists of a graph $G=(V,E)$ and an integer $k$. We have to decide if there exists a partition of $V$ to $L$ and $R$ such that $|L|=|R|=\frac{|V|}{2}$ where the number of edges $uv \in E$ with $u \in L$ and $v \in R$ is at most $k$. The constructed instance \tuple{G, A} of \GWO-\Schelling is on the same graph $G$ and it has $\frac{|V|}{2}$ red agents and $\frac{|V|}{2}$ blue agents. We will ask if there is an assignment \vbf such that the welfare of each group is at least $\frac{|V|}{2}-\frac{k}{3}$.

So, assume that there is a partition $V$ to $L$ and $R$ with $|L|=|R|=\frac{|V|}{2}$ such that there are exactly $\ell\le k$ edges between $L$ and $R$. We create an assignment \vbf by placing all blue agents on $L$ and all red agents on $R$. Observe that since $|A| = |V|$ and since the graph is cubic, for every $i$ we have that $f_i(\vbf) = 3-e_i(\vbf)$ and that $|N_i(\vbf)|=3$. Thus, for \(X\in \{R,B\}\)
\begin{align*}
    \sw_X(\vbf) & = \sum_{i \in X} \frac{f_i(\vbf)}{|N_i(\vbf)|} = \sum_{i \in X} \frac{3-e_i(\vbf)}{3} \\ & = \frac{|V|}{2}-\sum_{i \in X} \frac{e_i(\vbf)}{3}.
\end{align*}
Furthermore, observe that if $i$ is a blue agent, i.e. occupies a vertex in $L$, all his enemies occupy vertices in $R$ and vice versa. Hence, $\sum_{i \in B} e_i(\vbf) = \sum_{i \in R} e_i(\vbf) = \ell\le k$.
Thus, for \(X\in \{R,B\}\), $\sw_X(\vbf) = \frac{|V|}{2} - \frac{\ell}{3}\ge \frac{|V|}{2} - \frac{k}{3}$.

For the other direction, assume that we have an assignment \vbf such that for \(X\in \{R,B\}\) we have $\sw_X(\vbf)\ge \frac{|V|}{2} - \frac{k}{3}$. From the arguments above, we know that $\sw_R(\vbf) = \sw_B(\vbf) = \frac{|V|}{2}-\frac{\ell}{3}$, where \(\ell\) is the number of edges between vertices assigned to blue agents and the vertices assigned to red agents. Thus, $\ell \le k$ and there exists at most $k$ edges in $G$ where one of the endpoints is occupied by a red agent and the other endpoint is occupied by a blue agent. The proof is completed by creating a partition of $V$ by setting $L$ to contain all the vertices occupied by blue agents and $R$ to contain all the vertices occupied by red agents. It follows that there are at most $k$ edges between $L$ and $R$.
\end{proof}

Theorems~\ref{thm:rb-w-h} and~\ref{thm:bounded-NP-h} show that we cannot hope for an efficient algorithm, at least for \WO and \GWO, just by parameterizing only by $r+b$ or only by the maximum degree $\Delta$. We complement this by providing an FPT algorithm for $\phi$-\Schelling when parameterized by $r+b+\Delta$.
\iflong 
\begin{theorem}
\fi
\ifshort 
\begin{theorem}[$\star$]
\fi
\label{thm:fpt-rbD}
For every $\phi \in \{\WO, \PO, \UVO, \GWO\}$, $\phi$-\Schelling is in \FPT parameterized by $r+b+\Delta$.
Moreover, $\phi$-\Schelling admits a kernel with at most \(\mathcal{O}(\Delta^2\cdot r^2\cdot b^2)\) many vertices. 
\end{theorem}
\ifshort
\begin{proof}[Proof Sketch]
\fi 
\iflong 
\begin{proof}
\fi 
    Let \tuple{G,A} be a \Schelling instance. Let \(C_1,\ldots, C_k\) be the connected components of $G$ such that $|C_1|\ge |C_2|\ge \cdots \ge |C_k|$. We will prove the theorem in two steps. First we show that if $|C_1|\ge (\Delta + 1)\cdot r\cdot (1 + \Delta\cdot b)$, then we can construct an assignment $\vbf$ such that \(\sw(\vbf)=r+b\) in polynomial time. Afterwards, we show that there always exists a solution that maximizes social welfare that does not intersect any of the components $C_{r+b+1},\ldots, C_k$ and if $f(r+b+\Delta)\ge |C_1|\ge |C_2|\ge \ldots |C_{r+b}|$, we can find a solution that maximizes social welfare in \FPT time; for example by trying all possible assignments that assign all agents to the components $C_1, C_2,\ldots,C_{r+b}$. 
    
    Assume that $|C_1|\ge (\Delta + 1)\cdot r\cdot (1 + \Delta\cdot b)$ and let us pick an arbitrary set $X\subseteq C_1$ such that $G[X]$ is connected and $|X|=r$. 
    We assign all red agents to the vertices in $X$. Now \(|N[X]|\le (\Delta + 1)r\), where $N[X]$ is the closed neighborhood of $X$, and \(|N(N[X])|\le \Delta(\Delta + 1)r\). However, every connected component of $G[C_1\setminus N[X]]$ contains a vertex with a neighbor in $N[X]$. Hence, $G[C_1\setminus N[X]]$ has at most \(\Delta(\Delta + 1)r\) many connected components and if  $|C_1|\ge (\Delta + 1)r(1 + \Delta\cdot b)$ it follows that $|C_1\setminus N[N[X]]|\ge (\Delta + 1)\cdot r\cdot \Delta\cdot b$ and by the pigeonhole principle at least one of the connected components of $G[C_1\setminus N[X]]$ has $b$ vertices. We can assign all blue agents to a connected subgraph of such a component. Let $\vbf$ be the assignment we obtained above. Since no blue agent is assigned to \(N(X)\) and $X$ is connected, we get \(\sw_R(\vbf)=r\) whenever \(r\ge 2\) and since, in addition, all blue agents are also assigned to a connected subgraph of $G$, we get  \(\sw_B(\vbf)=b\) whenever \(b\ge 2\). Hence from now on we can assume that all connected components of $G$ have size at most \((\Delta + 1)\cdot r\cdot (1 + \Delta\cdot b)\). 
    
    It remains to show that we can focus on the $r+b$ largest components. 
    \ifshort
    We can observe that if agents of some type are assigned to some component $C_q$, $q > r+b$, then we can reassign all of them to some component $C_{p}$, $p\le r+b$, such that no vertex of $C_p$ is occupied by any agent so far. Moreover, this operation does not decrease the social welfare. 
    \fi
    \iflong Let \(\vbf\) be an assignment with the maximum social welfare and among all assignments that achieve the maximum social welfare, let \(\vbf\) be such assignment of agents that minimizes the maximum index of the connected component that contains a vertex that is assigned to some agent by \(\vbf\). That is \(\vbf\) minimizes \(\max \{q \mid v(i)\in C_q\text{ for some }i\in A\}\). We show by contradiction that \(\vbf\) assigns only agents in $C_1,\ldots, C_{r+b}$. 
    Let $q\in [k]$ be the maximum index such that $C_q$ is assigned some agents. If $q\le r+b$, we are done. Assume that $q>r+b$. Now let $R_q$ be all the red agents assigned to $C_q$ and $B_q$ be all the blue agents assigned to $C_q$. If $R_q$ is not empty, then let $p$ be the minimum index such that $C_p$ is not assigned any agent by $\vbf$. Since there are $r+b$ agents in total and $C_q$, $q> r+b$ is assigned some agents, it follows that $p\le r+b$. Since $|C_1|\ge |C_2|\ge \cdots \ge |C_k|$ and $|C_q|\ge |R_q|$, it follows that $|C_p|\ge |R_q|$ and we can create an assignment \(\vbf_1\) that assigns all agents in $R_q$ to a connected subgraph of $G[C_p]$ and all the remaining agents the same way as $\vbf$. It is readily seen that $\sw(\vbf_1)\ge \sw(\vbf)$. Similarly, if $B_q$ is not empty, we let $\ell$ to be the minimum index such that $C_\ell$ is not assigned any agent by $\vbf_1$ (if $R_q$ was empty then $\vbf_1=\vbf$). Note that since still the agents in $B_q$ are assigned to $C_q$, $q\ge r+b$, it follows that $\ell \le r+b$ and we can assign all agents to a connected subgraph of \(G[C_\ell]\). Let $\vbf_2$ be the final assignment. Again, it is easy to see that $\sw(\vbf_2)\ge \sw(\vbf_1)\ge \sw(\vbf)$, as the utilities of all agents outside $B_q$ do not change and the utilities of agents in $B_q$ can only increase. However, the maximum index of a connected component of $G$ that is assigned by $\vbf_2$ is strictly smaller than $q$. This is a contradiction with the choice of $\vbf$ and we conclude that $q\le r+b$. 
    
    \fi 
    Hence, we can remove all connected components $C_q$, $q>r+b$, from the instance and we are left with at most $r+b$ components, each with at most \(\mathcal{O}(\Delta^2\cdot r\cdot b)\) many vertices. This we can solve in time \(\mathcal{O}((\Delta^2\cdot r\cdot b\cdot (r+b))^{r+b})\) by enumerating all possible assignments of the agents. 
\end{proof}

\section{Multiple types}
In this section we depart from the standard model and we study Schelling instances with multiple types, denoted \SchellingM. We show that \perfect-\SchellingM is \W{1}-hard when parameterized by agent-types. Our reduction shows that the intractability \iflong of the problem \fi holds for a variety of graphs. In fact, it reveals that the sizes of the connected components of the graph are sufficient for proving hardness, without depending on the internal structure of every component. We will prove our result in two steps. In the first step we will prove that the problem is hard even when $G$ is a collection of connected components with arbitrary structure. Then, we will show how to get hardness even when $G$ is a tree.

\iflong
\begin{theorem}
\fi
\ifshort
\begin{theorem}[$\star$]
\fi
\label{thm:wh-k-types}
Let \(\mathcal{G}\) be an arbitrary class of connected graphs that contains at least one graph of size \(s\) for every \(s\in \mathbb{N}\). \perfect-\SchellingM is \NP-hard and \W{1}-hard when parameterized by agent-types, even when every connected component of $G$ is in \(\mathcal{G}\).
\end{theorem}
\iflong
\begin{proof}
\fi
\ifshort
\begin{proof}[Proof sketch]
\fi
We will prove our result via a reduction from \UBP. An instance of \UBP consists of a set $I$ of items, where every item $i \in I$ has a positive integer size $s_i$ given in unary, and $k$ bins of size $B$ each. The task is to decide if there is a partition of the items into $k$ subsets $I_1, I_2, \ldots, I_k$ such that the size of each subset is \iflong at most \fi \ifshort exactly \fi $B$\iflong , i.e., $\sum_{i \in I_j}s_i \leq B$ for every $j \in [k]$\fi . \UBP is \W{1}-hard parameterized by the number of bins $k$~\cite{JansenKMS13}. 
\iflong
In what follows, we will apply the following standard assumptions about the considered instances of \UBP:
\begin{itemize}
	\item $|I|> k$ (otherwise there is a trivial answer);
	\item $B$ and $s_i$ for every $i \in I$ are even integers (we can multiply $B$ and every $s_i$ without changing the answer to the decision question);
	\item $s_i \leq B-2$ for every item $i \in I$ (if there exists an item such that $s_i=B$, we remove it and decrease $k$ by one);
	\item $\sum_i s_i = k\cdot B$ (we can add dummy items of size 2 without changing the answer to the decision question).
\end{itemize} 
Hence the  goal is to decide whether there is a partition of the items into $k$ subsets such that the size of each subset~is~$B$.
\fi

Given an instance of \UBP, we will create an instance of \SchellingM with $k$ types of agents, where for each type there are $B$ agents, hence there are $k\cdot B$ agents in total. The graph $G$ will be the union of the graphs $G_1, G_2, \ldots, G_{|I|}$, where $G_i$ is isomorphic to a connected graph in \(\mathcal{G}\) with $s_i$ vertices. We will ask if there is a perfect assignment for \tuple{G,A}.

\ifshort
 The correctness is proven by observing that in any perfect partition the agents are neighbors only to agents of the same type, and hence any connected component contains only agents of a specific type.
 \fi
 \iflong
 So, assume that $(I_1, I_2, \ldots, I_k)$ is a solution of the instance of \UBP. We create an assignment \vbf for \tuple{G,A} as follows. If item $i \in I_j$, then in every vertex of $G_i$ we place an agent of type $j$. Observe that for every $j \in [k]$ we have that $\sum_{i \in I_j}s_i = B$; thus we have placed $B$ agents of each type on $G$. In addition, observe that for every $i \in I$ every graph $G_i$ has at least two vertices and all of its vertices are occupied by agents of the same type. Thus, every agent in every graph has only neighbors of the the same type and every agent has at least one neighbor since $G_i$ is connected. Hence, every agent gets utility 1 under \vbf and it is a perfect assignment.

 In order to prove the other direction, assume that there exists a perfect assignment \vbf for the instance \tuple{G,A} of \perfect-\SchellingM. Since \vbf is perfect, every agent has utility 1 under \vbf. This means for every type $j \in [k]$ and every agent $i$ of this type, $N_i(\vbf)$ contains only agents of type $j$. Thus, for every $i \in I$ the connected component $G_i$ contains only agents of the same type. So, we create a partition $I_1, I_2, \ldots, I_k$ as follows. If $G_i$ contains agents of type $j \in [k]$, then we put item $i$ in $I_j$. Recall that there are $B$ agents of type $j$ and all of them get utility 1 under \vbf. Thus, we have that $\sum_{i \in I_j}s_i = B$ for every $j \in [k]$, which is a solution for the original instance of \UBP.
\fi
\end{proof}

As a corollary of Theorem~\ref{thm:wh-k-types} we can prove that the problem remains \W{1}-hard even when $G$ is a tree.
\iflong
\begin{corollary}
\fi
\ifshort
\begin{corollary}[$\star$]
\fi
\label{cor:wh-types-tree}
\perfect-\SchellingM is \NP-hard and \W{1}-hard when parameterized by agent-types, even if $G$ is a tree.
\end{corollary}
\iflong
\begin{proof}
Our starting point is the construction from the proof of Theorem~\ref{thm:wh-k-types}. This time though, we set every connected component $G_i$ to be a path. In addition, we add a vertex $v_0$ and we connect it with one endpoint of every path $G_i$. The correctness of the construction is verbatim to the one of Theorem~\ref{thm:wh-k-types}. The soundness follows from the observation that in any assignment \vbf with social welfare $k\cdot B$ vertex $v_0$ should be empty. If it is not empty, then there exists at least one agent with agents of different type in his neighborhood. Thus his utility under \vbf is strictly less than 1.
\end{proof}
\fi

Again, using Proposition~\ref{pro:notions} and Observation~\ref{obs:perfect}, we can get the following corollary.
\begin{corollary}
\label{cor:trees-notper}
For every $\phi \in \{\WO, \PO, \UVO, \GWO\}$, $\phi$-\SchellingM is \NP-hard and \W{1}-hard when parameterized by agent-types, even when $G$ is a tree.
\end{corollary}

In the rest of this section, we give an algorithm for $\phi$-\SchellingM that matches the lower bound from Corollary~\ref{cor:trees-notper}. Namely, we give an XP algorithm for the problem parameterized by the number of agent-types and the treewidth of the graph. Recall, that trees have treewidth one and hence \SchellingM does not admit an FPT algorithm parameterized by agent-types and treewidth unless $\FPT=\W{1}$. Hence, the rest of the section is devoted to the proof of the following theorem. 

\begin{theorem}\label{thm:tw_algorithm}
There is an \(|A|^{\bigoh(k\cdot\tw(G))}\cdot |V(G)|\) time algorithm for $\phi$-\SchellingM, $\phi \in \{\WO, \PO, \GWO, \UVO\}$, where \(k\) is the number of agent-types.
\end{theorem}

\noindent
Let \tuple{G,A} be an instance of $\phi$-\SchellingM with \(k\) agent-types, where \(G\) has treewidth at most \(w\). For a type \(i\in [k]\), we let \(A_i\) denote the set of all agents of type \(i\). Moreover, let \(\mathcal{T}=(T,\chi)\) be a tree-decomposition of \(G\) of width at most \(w\). Note that every assignment that is \WO is also $\phi$ by Proposition~\ref{pro:notions}, hence regardless of \(\phi\) we can compute a \WO assignment. 
The algorithm is a standard bottom-up dynamic programming along a nice tree-decomposition.
As always, the main challenge is to decide what records we should keep for
each node $t$ of $T$, where each record models an equivalence class of
partial assignments for the sub-instance induced by the vertices in $G_t$,
i.e., the graph $G$ induced by all vertices contained in bags in the subtree rooted at $t$. 

Consider some node \(t\in V(T)\). We would like to compute a table \(\Gamma_t\), where each entry of the table corresponds to some equivalence class 
of partial assignments and the value of that entry is the "best" partial assignment in the equivalence class.

For the algorithm to be efficient, we need the number of equivalence classes to be small and we should be able to compute each entry in \(\Gamma_t\) efficiently from the tables for the children of \(t\). First let us formally define a partial assignment. A \emph{partial assignment over the subset of agents $A'\subseteq A$} is an assignments \(\vbf_{A'} = (v(a_1),v(a_2),\ldots, v(a_{|A'|}))\), where $A' = \{a_1,a_2, \ldots, a_{|A'|}\}$. A \emph{partial assignment} is then an assignment over some subset of agents. In the node $t$, we are interested in the partial assignment that we can obtain by taking a (full) assignment and restricting it to the agents that are assigned some vertex in $G_t$. For a partial assignment \(\vbf\) we let \(\sw^t(\vbf)=(\sum_{i\in A'}u_i(\vbf))\), where \(A'\subseteq A\) is the set of agents assigned to vertices in \(G_t-\chi(t)\) by \(\vbf\). Note that \(\sw^t(\vbf)\) does not contain the utilities of vertices in \(\chi\), that is because the utilities of these vertices still depend on the vertices in \(V(G)\setminus V(G_t)\) and might be different in an assignment \(\vbf'\) whose restriction to \(G_t\) results in \(\vbf\). 

A \emph{(description of) equivalence class } \(\mathcal{C}\) for the node $t$ is a tuple \(\tuple{\Sizes,\Types,\NeighborhoodSizes}\), where 
\begin{itemize}
    \item \(\Sizes\colon [k]\rightarrow \mathbb{N}\) such that \(0\le \Sizes(i)\le |A_i|\) for all \(i\in [k]\),
    \item \(\Types\colon\chi(t)\rightarrow \{0,\ldots, k\}\), and 
    \item \(\NeighborhoodSizes\colon \chi(t)\rightarrow \mathbb{N}^k\) such that for every vertex \(v\in \chi(t)\), if \(\NeighborhoodSizes(v) = \tuple{n_1,n_2,\ldots, n_k}\), then \(0\le n_i\le |A_i|\) for all \(i\in [k]\).
\end{itemize}

We say that a partial assignment \(\vbf\) \emph{belongs} to the equivalence \(\mathcal{C}=\tuple{\Sizes,\Types,\NeighborhoodSizes}\) for the node $t$ if and only if:
\begin{itemize}
    \item \(\vbf\) assigns agents only to vertices in \(G_t\),
    \item for all \(i\in [k]\), the number of agents of type \(i\) assigned a vertex in $G_t$ by \(\vbf\) is \(\Sizes(i)\), 
    \item for all \(v\in \chi(t)\), if \(\Types(v) = 0\), then \(\vbf\) does not assign any agent to the vertex \(v\), else \(\vbf\) assigns some agent of type \(\Types(v)\) to \(v\), and
    \item for all \(v\in \chi(t)\) and all \(i\in [k]\), the number of neighbors of $v$ in $G_t$ assigned an agent of type $i$ is equal to \(\NeighborhoodSizes(v)[i]\).
\end{itemize}
It is easy to see that every partial assignment belongs to some equivalence class. 

We say that an equivalence class is \emph{valid} if there exists a partial assignment that belongs to the equivalence class. 
For a valid equivalence class \(\mathcal{C}\), the table entry \(\Gamma_t[\mathcal{C}]\) should contain
some partial assignment \(\vbf\) that belongs to \(\mathcal{C}\) and maximizes $\sw^t(\vbf)$.

Moreover, recall that for the root node $r$ of $T$, we have $G_r=G$ and \(\chi_v=\emptyset\), hence $\sw^r(\vbf)= \sw(\vbf)$ and the table entry \(\Gamma_r[\mathcal{C}_r]\) for the class \(\mathcal{C}_r=\tuple{\Sizes,\Types,\NeighborhoodSizes}\), 
where for all \(i\in [k]\) we have \(\Sizes(i)=|A_i|\) and \(\Types\) and \(\NeighborhoodSizes\) are empty functions, contains an assignment that maximizes 
the social welfare. First, let us observe that the number of entries in each node is bounded. 

\iflong
\begin{observation}
\fi 
\ifshort 
\begin{observation}[$\star$]
\fi 
\label{obs:No_eq_classes}
The number of equivalence classes for node $t$, $t\in V(T)$, is at most \((|A|+1)^{k(1+|\chi(t)|)}\cdot (k+1)^{|\chi(t)|}\).
\end{observation}
\iflong 
\begin{proof}
    For every \(i\in [k]\) there are \(|A_i|+1\) possibilities for the value of \(\Sizes(i)\), hence there are at most \(\prod_{i\in k}(|A_i|+1)\le (|A|+1)^k\) many possibilities for the function \(\Sizes\). Moreover, for each vertex in \(\chi(v)\), we have at most $k+1$ possibilities for \(\Types(i)\) and at most \(\prod_{i\in k}(|A_i|+1)\le (|A|+1)^k\) many possibilities for \(\NeighborhoodSizes(i)\) and the observation follows. 
\end{proof}
\fi 

It follows from Observation~\ref{obs:No_eq_classes} that it suffices to show that for each node $t$ and each equivalence class \(\mathcal{C}\) for the node $t$, we can in time \(|A|^{\bigoh(k\cdot\tw(G))}\) decide whether \(\mathcal{C}\) is valid and if so find the partial assignment \(\vbf\) that belongs to \(\mathcal{C}\) and maximizes \(\sw^t(\vbf)\).
\ifshort
\begin{lemma}[$\star$]\label{lem:all_nodes}
Let $t\in V(T)$ be a node 
and \(\mathcal{C} = \tuple{\Sizes,\Types,\NeighborhoodSizes}\) an equivalence class for $t$.
Then, assuming that we computed all entries for the children of $t$ in $T$,  we can in \(|A|^{\bigoh(k\cdot |\chi(t)|)}\) time decide whether \(\mathcal{C}\) is valid and if so compute a partial assignment \(\vbf\) that belongs to \(\mathcal{C}\) and maximizes \(\sw^t(\vbf)\). 
\end{lemma}
\fi
\iflong 
We will distinguish four cases depending on the type of the node $t$. Moreover, when computing the entries for the node $t$ we always assume that we computed all entries for the children of $t$ in $T$. 
\fi 
\iflong 
\iflong 
\begin{lemma}[leaf node]
\fi
\ifshort 
\begin{lemma}[leaf node ($\star$)]
\fi
\label{lem:leaf_node}
Let $t\in V(T)$ be a leaf node and \(\mathcal{C} = \tuple{\Sizes,\Types,\NeighborhoodSizes}\) an equivalence class for $t$. 
Then we can in \(\mathcal{O}(k)\) time decide whether \(\mathcal{C}\) is valid and if so compute a partial assignment \(\vbf\) that belongs to \(\mathcal{C}\) and maximizes \(\sw^t(\vbf)\). 
\end{lemma}
\iflong
\begin{proof}
 By definition of nice tree-decomposition, we have that \(\chi(t)=V(G_t)=\emptyset\). Hence the only valid assignment for \(t\) is the empty assignment that results in the 
 equivalence class \(\mathcal{C}' = \tuple{\Sizes',\Types',\NeighborhoodSizes'}\), where for all $i\in[k]$ we have $\Sizes'(i)=0$ and \(\Types\) and \(\NeighborhoodSizes\) are empty functions. It is easy to see that we can in time \(\mathcal{O}(k)\) decide whether \(\mathcal{C}= \mathcal{C}'\) and if so output the empty assignment.
\end{proof}
\fi 

\iflong 
\begin{lemma}[introduce node]
\fi
\ifshort 
\begin{lemma}[introduce node ($\star$)]
\fi
\label{lem:introduce_node}
Let $t\in V(T)$ be an introduce node with child $t'$ such that $\chi(t)\setminus \chi(t')=\{v\}$ and \(\mathcal{C} = \tuple{\Sizes,\Types,\NeighborhoodSizes}\) an equivalence class for $t$. 
Then we can in polynomial time decide whether \(\mathcal{C}\) is valid and if so compute a partial assignment \(\vbf\) that belongs to \(\mathcal{C}\) and maximizes \(\sw^t(\vbf)\). 
\end{lemma}

\iflong 
\begin{proof}

First note that from the properties of a tree-decomposition it follows that all neighbors of \(v\) are in \(\chi(v)\). Hence, if for some $i\in [k]$ we have that \(\NeighborhoodSizes(v)[i]\) does not equal to the number of neighbors $w$ of $v$ such that \(\Types(w)=i\), then there cannot exist a partial assignment that belongs to \(\mathcal{C}\) and \(\mathcal{C}\) is not valid. 
Else, let \(\mathcal{C}' = \tuple{\Sizes',\Types',\NeighborhoodSizes'}\) be an equivalence class for \(t'\) such that
\begin{itemize}
    \item for all \(i\in[k]\), \(\Sizes'(i) = \Sizes(i)- 1\) if \(\Types(v)=i\) and \(\Sizes'(i) = \Sizes(i)\) otherwise, 
    \item for all \(w\in \chi(t')=\chi(t)\setminus \{v\}\), \(\Types'(w) = \Types(w)\), and 
    \item for all \(i\in[k]\) and \(w\in \chi(t')\), \(\NeighborhoodSizes'(w)[i]= \NeighborhoodSizes(w)[i] - 1\) if \(vw\in E(G)\) and \(\Types(v)=i\); and \(\NeighborhoodSizes'(w)[i]= \NeighborhoodSizes(w)[i]\) otherwise. 
\end{itemize}
It is easy to see that if a partial assignment \(\vbf\) belongs to \(\mathcal{C}\), then the partial assignment \(\vbf'\) obtained from \(\vbf\) by restriction to \(G_{t'}\) belongs to \(\mathcal{C}'\). Similarly, if a partial assignment \(\vbf'\) belongs to \(\mathcal{C}'\), then the partial assignment \(\vbf\) we obtain from \(\vbf'\) by assigning the vertex $v$ an agent of type \(\Types(v)\) belongs to \(\mathcal{C}\). Moreover, for such pair of partial assignments \(\vbf\) and \(\vbf'\), the utility of every agent in $V(G_t)\setminus \chi(t) = V(G_{t'})\setminus \chi(t') $ is the same. Therefore, \(\sw^t(\vbf) = \sw^{t'}(\vbf')\). It follows that \(\mathcal{C}\) is valid if and only if \(\mathcal{C}'\) is valid. Moreover, if \(\Gamma_{t'}[\mathcal{C}']=\vbf'\), then we obtain \(\Gamma_{t}[\mathcal{C}]\) by taking \(\vbf'\) and in addition assigning an arbitrary agent of type \(\Types(v)\) the vertex $v$. 
\end{proof}
\fi 

\iflong 
\begin{lemma}[forget node]
\fi
\ifshort 
\begin{lemma}[forget node ($\star$)]
\fi
\label{lem:forget_node}
Let $t\in V(T)$ be a forget node with child $t'$ such that $\chi(t')\setminus \chi(t)=\{v\}$ and \(\mathcal{C} = \tuple{\Sizes,\Types,\NeighborhoodSizes}\) an equivalence class for $t$. 
Then we can in \(\mathcal{O}((k+1)|A|^k)\) time decide whether \(\mathcal{C}\) is valid and if so compute a partial assignment \(\vbf\) that belongs to \(\mathcal{C}\) and maximizes \(\sw^t(\vbf)\). 
\end{lemma}

\iflong 
\begin{proof}
Note that \(G_t = G_{t'}\). Hence a partial assignment \(\vbf\) assigns agents only to the vertices of \(G_t\) if and only if it assigns agents only to vertices of \(G_{t'}\). Now let \(\vbf\) be a partial assignment that belongs to \(\mathcal{C}\) and let us to which equivalence classes for $t'$ it can belong. Let \(\vbf\) belong to the equivalence class \(\mathcal{C}'= \tuple{\Sizes',\Types',\NeighborhoodSizes'}\) for $t'$.
Then \(\Sizes = \Sizes'\) and for all $w\in \chi(t)$ we have \(\Types(w) = \Types'(w)\) and  \(\NeighborhoodSizes(w) = \NeighborhoodSizes'(w)\). However, \(\Types(v)\) and \(\NeighborhoodSizes(v)\) is not defined, since $v$ is not in $\chi(t)$, hence \(\Types'(v)\) and \(\NeighborhoodSizes'(v)\) can be arbitrary. 

Let $\mathrm{C}$ be the set of all valid equivalence classes for $t'$ such that for all $\mathcal{C}'= \tuple{\Sizes',\Types',\NeighborhoodSizes'} \in \mathrm{C}$ it holds that:
\begin{itemize}
    \item  \(\Sizes = \Sizes'\), and 
    \item for all $w\in \chi(t)$, \(\Types(w) = \Types'(w)\) and  \(\NeighborhoodSizes(w) = \NeighborhoodSizes'(w)\).
\end{itemize}
As discussed above, for every partial assignment \(\vbf\) that belongs to \(\mathcal{C}\), there exists an equivalence class \(\mathcal{C'}\in \mathrm{C}\) such that \(\vbf\) belongs to \(\mathcal{C}'\).  Hence, if \(\mathrm{C}\) is empty, we can return that \(\mathcal{C}\) is not valid.  Moreover, it is easy to see that a partial assignment that belongs to some class in \(\mathrm{C}\) also belongs to \(\mathcal{C}\). 

Let \(\mathcal{C}'\) be an equivalence class in \(\mathrm{C}\), and let $u^{\mathcal{C}'}(v)=0$ if $\Types'(v)=0$ and \[u^{\mathcal{C}'}(v)=\frac{\NeighborhoodSizes'(v)[\Types'(v)]}{\sum_{i\in [k]}\NeighborhoodSizes'(v)[i]}\] otherwise. 

For a partial assignment \(\vbf\) that belongs to \(\mathcal{C}'\), we have that \(\sw^t(\vbf)=\sw^{t'}(\vbf) + u^{\mathcal{C}'}(v)\). Since $u^{\mathcal{C}'}(v)$ is constant for given equivalence class $\mathcal{C}'$, it follows that if \(\vbf\) is a partial assignment that belongs to \(\mathcal{C}\) and maximizes \(\sw^t(\vbf)\), then \(\vbf\) also maximizes $\sw^{t'}(\vbf)$. Moreover, if \(\vbf\) belongs to \(\mathcal{C}'\in \mathrm{C}\), then all partial assignments \(\vbf'\) that maximize $\sw^{t'}(\vbf')$ also maximize \(\sw^t(\vbf)\). 

Therefore, to compute \(\Gamma_t[\mathcal{C}]\) we compute for every \(\mathcal{C}'\in \mathrm{C}\) the value \(\mathrm{val}_{\mathcal{C}'}=\sw^{t'}(\Gamma_{t'}[\mathcal{C}']) + u^{\mathcal{C}'}(v)\}\) and we let \(\Gamma_t[\mathcal{C}]= \Gamma_{t'}[\mathcal{C'}]\) for arbitrary equivalence class \(\mathcal{C}'\in \mathrm{C}\) that maximizes \(\mathrm{val}_{\mathcal{C}'}\). The running time then follows from the observation that \(|\mathrm{C}|\le (k+1)|A|^k\). 
\end{proof}
\fi

\iflong 
\begin{lemma}[join node]
\fi
\ifshort 
\begin{lemma}[join node ($\star$)]
\fi
\label{lem:join_node}
Let $t\in V(T)$ be a join node with children $t_1$ and $t_2$, where $\chi(t)=\chi(t_1)=\chi(t_2)$ and \(\mathcal{C} = \tuple{\Sizes,\Types,\NeighborhoodSizes}\) an equivalence class for $t$.
Then we can in \((\bigoh(|A|+1)^{k(1+|\chi(t)|)})\) time decide whether \(\mathcal{C}\) is valid and if so compute a partial assignment \(\vbf\) that belongs to \(\mathcal{C}\) and maximizes \(\sw^t(\vbf)\). 
\end{lemma}
\iflong 
\begin{proof}
For \(i\in [k]\), let $S_i$ be the set of vertices $v$ in \(\chi(t)\) such that \(\Types(v)=i\). Moreover for a vertex $v\in \chi(t)$ and \(i\in [k]\), let \(N^i_{\chi(t)}(v) = \{w\mid vw\in E(G)\land w\in \chi(t)\land \Types(w)=i\}\); that is the set of neighbors $w$ of \(v\) in \(\chi(t)\) such that \(\Types(w)=i\).  

First, assume that \(\mathcal{C}\) is valid and let \(\vbf\) be a partial assignment that belongs to \(\mathcal{C}\). Note that for every agent-type \(i\in[k]\), there are precisely \(\Sizes(i)-|S_i|\) many agents assigned a vertex in \(G_t\setminus \chi(t)\). Every such vertex is either in \(G_{t_1}\) or in \(G_{t_2}\) but not both. Similarly for every \(i\in[k]\) and a vertex \(v\in \chi(v)\) the agents of type \(i\) are assigned to precisely \(\NeighborhoodSizes(v)[i]- |N^i_{\chi(t)}(v)|\) many neighbors of \(v\) in \(G_t\setminus \chi(t)\). 
Now, let \(\vbf_1\) and \(\vbf_2\) be the restrictions of \(\vbf\) to \(G_{t_1}\) and \(G_{t_2}\), respectively. Finally, let \(\mathcal{C}_1 = \tuple{\Sizes_1,\Types_1,\NeighborhoodSizes_1}\) and \(\mathcal{C}_2 = \tuple{\Sizes_2,\Types_2,\NeighborhoodSizes_2}\) be the equivalence classes such that \(\vbf_1\) belongs to \(\mathcal{C}_1\) and \(\vbf_2\) belongs to \(\mathcal{C}_2\).
First, as \(\chi(t)=\chi(t_1)=\chi(t_2) \) it follows that \(\Types = \Types_1 = \Types_2\). Notice that \(\sw^t(\vbf) = \sw^{t_1}(\vbf_1) + \sw^{t_1}(\vbf_2)\).  Moreover, from our discussion above it follows that, for every \(i\in [k]\), \(\Sizes(i)-|S_i| = \Sizes_1(i)-|S_i| + \Sizes_2(i)-|S_i| \), or equivalently \(\Sizes(i) = \Sizes_1(i) + \Sizes_2(i)-|S_i|\). Similarly for every vertex \(v\in \chi(t)\) and every \(i\in[k]\) we have \(\NeighborhoodSizes(v)[i] = \NeighborhoodSizes_1(v)[i] + \NeighborhoodSizes_2(v)[i]-|N^i_{\chi(t)}(v)|\). 

On the other hand, assume that we have two valid equivalence classes \(\mathcal{C}_1 = \tuple{\Sizes_1,\Types_1,\NeighborhoodSizes_1}\) for \(t_1\) and \(\mathcal{C}_2 = \tuple{\Sizes_2,\Types_2,\NeighborhoodSizes_2}\) such that 

\begin{enumerate}
    \item[(A)] \(\Types = \Types_1 = \Types_2\); 
    \item[(B)] for all \(i\in [k]\), \(\Sizes(i) = \Sizes_1(i) + \Sizes_2(i)-|S_i|\); and
    \item[(C)]  for all \(v\in \chi(t)\) and all \(i\in [k]\), \(\NeighborhoodSizes(v)[i] = \NeighborhoodSizes_1(v)[i] + \NeighborhoodSizes_2(v)[i]-|N^i_{\chi(t)}(v)|\) .
\end{enumerate}
Let \(\vbf_1\) be a partial assignment that belongs to \(\mathcal{C}_1\) and \(\vbf_2\) be a partial assignment that belongs to \(\mathcal{C}_2\). We can construct a partial assignment \(\vbf\) that belongs to \(\mathcal{C}\) by assigning each vertex \(v\in \chi(t)\) to an agent of type \(\Types(v)\), every vertex \(w\in G_{t_1}\setminus\chi(t)\) an agent of the same type as the agent assigned the vertex \(w\) by \(\vbf_1\), and every vertex \(u\in G_{t_2}\setminus\chi(t)\) an agent of the same type as the agent assigned the vertex \(u\) by \(\vbf_2\). It is straightforward to verify that \(\vbf\) belongs to \(\mathcal{C}\) and that \(\sw^t(\vbf) = \sw^{t_1}(\vbf_1) + \sw^{t_1}(\vbf_2)\). 

From the above discussion, it follows that it suffices to go over all pairs of 
equivalence classes \(\mathcal{C}_1\) and \(\mathcal{C}_2\) that satisfy the conditions 
(A), (B), and (C). For each pair we compute \(\mathrm{val}_{\mathcal{C}_1,\mathcal{C}_2 }= \sw^{t_1}(\Gamma_{t_1}[\mathcal{C}_1])+\sw^{t_2}(\Gamma_{t_2}[\mathcal{C}_2])\) and 
pick arbitrary pair \(\mathcal{C}_1\) and \(\mathcal{C}_2\) that maximizes  
\(\mathrm{val}_{\mathcal{C}_1,\mathcal{C}_2}\). We can then compute the assignment 
\(\Gamma_t[\mathcal{C}]\) from \(\Gamma_{t_1}[\mathcal{C}_1]\) and 
\(\Gamma_{t_2}[\mathcal{C}_2]\) as described before. The running time follows from 
the observation that once we fix the equivalence class \(\mathcal{C}_1\), then the 
equivalence class \(\mathcal{C}_2\) is fixed by the conditions (A), (B), and (C). 
Moreover, \(\Types_1\) and \(\Types_2\) is fixed by the condition (A). 
\end{proof}
\fi 
\fi
\iflong
We are now ready to prove the main theorem. 
\fi

\begin{proof}[Proof of~Theorem~\ref{thm:tw_algorithm}]
Let $\tuple{G,A}$ be an instance of $\phi$-\SchellingM, with $\phi \in \{\WO, \PO, \GWO, \UVO\}$. We will compute a \WO assignment, which by Proposition~\ref{pro:notions} is \PO, \GWO, and \UVO. The algorithm first \iflong uses Proposition~\ref{fact:findtw} to compute \fi\ifshort computes \fi{} a nice tree decomposition \(\mathcal{T}=(T,\chi)\) of $G$ of width  \iflong at most $w$ in FPT-time w.r.t. $w$\fi\ifshort \(w\le 5\tw(G)+4\) in FPT-time~\cite{BodlaenderDDFLP16}\fi. \iflong Note that $w$ is at most \(5\tw(G)+4\).\fi  Afterwards, we use the \iflong algorithms of Lemmas~\ref{lem:leaf_node}, \ref{lem:introduce_node}, \ref{lem:forget_node}, and \ref{lem:join_node} \fi\ifshort algorithm of Lemma~\ref{lem:all_nodes} \fi to compute for every node $t$ and every valid equivalence class \(\mathcal{C}\) for \(t\) a partial assignment that belongs to \(\mathcal{C}\) and maximizes \(\sw^t(\vbf)\) among all the partial assignments that belong to the equivalence class \(\mathcal{C}\). An assignment that maximizes the social welfare is then the partial assignment that we computed for the root node of $T$ for the equivalence class 
\(\mathcal{C}_r=\tuple{\Sizes,\Types,\NeighborhoodSizes}\), 
where for all \(i\in [k]\) we have \(\Sizes(i)=|A_i|\) and \(\Types\) and \(\NeighborhoodSizes\) are empty functions. The correctness follows from the correctness of \iflong Lemmas~\ref{lem:leaf_node}, \ref{lem:introduce_node}, \ref{lem:forget_node}, and \ref{lem:join_node}\fi\ifshort Lemma~\ref{lem:all_nodes}\fi . The running time of the algorithm is at most the number of nodes of $T$, i.e., at most $w^2|V(G)|$, times the maximum number of equivalence classes for a node in $t$, i.e.,  \((|A|+1)^{k(2+w)}\cdot (k+1)^{w+1}\) by Observation~\ref{obs:No_eq_classes}, times the maximum time required to compute a partial assignment for a node $t$ and an equivalence class \(\mathcal{C}\) for any of the four node types of a nice tree-decomposition, which because of \iflong Lemmas~\ref{lem:leaf_node}, \ref{lem:introduce_node}, \ref{lem:forget_node}, and \ref{lem:join_node} is at most \((\bigoh(|A|+1)^{k(2+w)})\)\fi\ifshort Lemma~\ref{lem:all_nodes} is at most \(|A|^{\bigoh(k w)}\)\fi . Therefore, \(|A|^{\bigoh(kw)}\cdot |V(G)|= |A|^{\bigoh(k\cdot \tw(G))}\cdot |V(G)|\) is the total running time of the algorithm. 
\end{proof}

Observe that the number of agent-types is always at most the number of agents. It then follows from the running time of the algorithm that \SchellingM is actually FPT when parameterized by treewidth plus number of agents. 
\begin{corollary}
\label{cor:fpt-tw-n}
$\phi$-\SchellingM is in \FPT{} when parameterized by treewidth and the number of agents, for every $\phi \in \{\WO, \PO, \GWO, \UVO\}$. 
\end{corollary}

Finally, while Corollary~\ref{cor:trees-notper} implies that we cannot obtain an FPT algorithm if the number of agent-types is part of the parameter, unless $\FPT=\W{1}$, it remains an interesting open question whether this is possible for a constant number of agent types. 
\ifshort
While we do not fully resolve this question, we show that if we are looking for a perfect assignment, then we can modify \(\NeighborhoodSizes\) to keep only a Boolean information for every agent encoding whether it already has a neighbor of the same type. If we follow the proof of Theorem~\ref{thm:tw_algorithm} with this modification and we also reject all partial assignments that cannot lead to a perfect assignment, we get the following result.
\fi
\iflong While our algorithm cannot resolve this in general, we would like to point out a specific case that can be solved in FPT-time with a very minor modification of our algorithm. Assume that we wish to maximize the social welfare under the additional constraint that every agent is allowed to have only neighbors of the same type. Note that this is the only way to get social welfare equal to the number of agents. In this case, we would reject any equivalence class \(\mathcal{C} = \tuple{\Sizes,\Types,\NeighborhoodSizes}\), where \(\NeighborhoodSizes[v][i]>0\) such that $\Types[v]\neq i$. Moreover, if   $\Types[v]= i$, then we only care whether \(\NeighborhoodSizes[v][i]>0\) and the utility of the agent that is assigned vertex $v$ is $1$ or \(\NeighborhoodSizes[v][i]=0\) and the utility of this agent is $0$. Hence, we can replace \(\NeighborhoodSizes\) by a function from \(\chi(t)\rightarrow \{0,1\}\) with meaning that if \(\NeighborhoodSizes(v)=0\), then no neighbor of $v$ is assigned to any agent and if \(\NeighborhoodSizes(v)=1\), then at least one of the neighbors of $v$ is assigned some agent of \(\Types(v)\) and no neighbor of $v$ is assigned a agent of any other type. The algorithm then follows more or less analogously the proof of Theorem~\ref{thm:tw_algorithm}. It is easy to see that the number of these "modified" equivalence classes is at most $|A|^k\cdot 2^{|\chi(t)|}\cdot(k+1)^{|\chi(t)|}$, which is FPT by treewidth if the number of agent-types $k$ is a fixed constant. 
\fi 
\iflong
\begin{corollary}
\fi 
\ifshort 
\begin{corollary}[$\star$]
\fi 
\label{cor:fpt-perfect}
When the number of types is constant, \perfect-\SchellingM admits an FPT algorithm parameterized by treewidth. 
\end{corollary}

\section{Conclusions}
\label{sec:conclusions}

In this paper we studied $\phi$-\Schelling for $\phi \in \{\WO, \PO, \UVO, \GWO\}$. We presented both strong negative results and accompanying algorithms.
Our results show that tractability of $\phi$-\Schelling for every optimality notion considered {\em requires} the underlying graph to be constrained. Indeed, Theorem~\ref{thm:fpt-rbD} shows that the problem is fixed parameter tractable when parameterized by $r+b+\Delta$, where $\Delta$ is the maximum degree of the graph. One immediate question is whether we can strengthen this result by removing parameter $r$. 
Another intriguing and challenging question is to study $\phi$-\Schelling under the treewidth parameter. In this case, is there an FPT algorithm or is it \W{1}-hard? We conjecture that it is the latter case. As an intermediate step, one could parameterize by the vertex cover of the underlying graph, a much stronger parameter compared to treewidth. It is not too hard to show that when the number of agents equals the number of vertices of the graph the problem is fixed parameter tractable when parameterized by vertex cover. However, the problem remains challenging if the number of agents is less than the number of vertices of the graph. Is it in \FPT when parameterized by vertex cover, or is it \W{1}-hard? We believe that novel algorithmic techniques are required in order to answer the questions above.

\bibliographystyle{named}
\bibliography{bibliography}
\end{document}